\documentclass[draftclsnofoot,onecolumn]{IEEEtran}\usepackage{epsfig,latexsym,amsmath,amssymb,setspace,cite,color,graphicx,algorithm,algorithmic,cases}
\usepackage[font=footnotesize]{subfig}
\usepackage{tikz}

\usepackage{amsmath}
\usetikzlibrary{positioning, arrows.meta, calc, decorations.pathreplacing}

\usepackage[percent]{overpic}
\usepackage{amsmath}
\usepackage{amsfonts}
\usepackage{amssymb}
\usepackage{dsfont}
\usepackage{bm}
\usepackage{amsthm}
\usepackage{newlfont}
\usepackage{float}
\usepackage{hyperref}
\usepackage{algorithm}
\usepackage{algorithmic}
\usepackage{enumerate}
\usepackage{chngcntr}
\usepackage{mathtools}
\usepackage{breqn}
\usepackage{stmaryrd}
\usepackage{cite}
\usepackage{amsmath}
\usepackage{amsfonts}
\usepackage{amssymb}
\usepackage{dsfont}
\usepackage{bbm}
\usepackage{amsthm}
\usepackage{newlfont}
\usepackage{float}
\usepackage{hyperref}
\usepackage{enumerate}
\usepackage{chngcntr}
\usepackage{mathtools}
\usepackage{breqn}
\usepackage{stmaryrd}
\usepackage{cite}
\usepackage{algorithm}
\usepackage{algorithmic}
\usepackage{caption}
\hypersetup{
    colorlinks=true, 
    linkcolor= blue,
    citecolor=blue 
}
\newtheorem{definition}{Definition}
\newtheorem{lem}{Lemma}
\newtheorem{thm}{Theorem}

\usepackage{multicol}
\usepackage{cite}
\usepackage{amsmath,amssymb,amsfonts}
\usepackage{algorithmic}
\usepackage{graphicx}
\usepackage{textcomp}
\usepackage{xcolor}
\def\BibTeX{{\rm B\kern-.05em{\sc i\kern-.025em b}\kern-.08em
    T\kern-.1667em\lower.7ex\hbox{E}\kern-.125emX}}
    \IEEEoverridecommandlockouts
\allowdisplaybreaks
\title{Communication and Privacy Trade-Offs in Distributed Storage  With Shared Secret Keys}

\author{
 Maryam Keshvari, R\'emi A. Chou \thanks{M. Keshvari is with the School of Computing at Wichita State University. R. Chou is with the CSE Department at The University of Texas at Arlington.  This work was supported in part by NSF grant CCF-2425371. Part of the results has been presented in~\cite{keshvari2023distributed}.
 E-mail: remi.chou@uta.edu. }
}
\begin{document}
\maketitle

\begin{abstract}
We study distributed storage over public communication channels where privacy is realized using shared secret keys. Specifically, a user wishes to store a file across $L$ servers such that any subset of at least $t$ servers can reconstruct the file, while any subset of at most $z$ colluding servers learns no more than a fraction $\alpha$ of the file. Unlike classical secret-sharing models that assume secure communication channels between the user and each server, we consider a setting where all transmissions occur over a public channel and the user shares an individual secret key of length $n$ bits with each server. For a fixed key length $n$, we characterize the trade-off between the privacy leakage parameter $\alpha$ and the maximum file size that can be securely stored. Additionally, we determine the minimum local randomness, public communication, and storage required to achieve this optimal trade-off.
\end{abstract}
\begin{IEEEkeywords}
Distributed storage, colluding servers, secret-key–assisted communication, privacy--communication trade-off.
\end{IEEEkeywords}

\section{Introduction}
 Centralizing sensitive information storage puts the entire database at risk of being compromised, in the case of a data breach. Adopting a decentralized storage strategy can provide resilience against data breaches and avoid having a single point of entry for hackers.  For instance, secret sharing \cite{b2,b3} offers an effective solution to securely store data among $L$ servers, such that any coalition of $t\leq L$ servers can reconstruct the file by pooling their information but  any coalition of at most $ t-1 $ compromised servers cannot learn any information, in an information-theoretic sense, about the file. Applications of secret sharing to secure distributed storage have been extensively studied in a variety of contexts \cite{b4,b5,b6,b7,b8,b9,b10,b11,b12}. As motivated in\cite{b14,b15,b16}, the servers could be distinct cloud storage service providers, since companies may outsource data storage to reduce operating costs.

In this paper, we study secure distributed storage where communication between the user and the servers takes place over public channels, and privacy is enforced using shared secret keys. Specifically, we consider a system where a user wishes to store a file across $L$  servers such that $(i)$ any coalition of $t \leq L $ servers who combine their information can reconstruct the file, and $(ii)$ any $z$ colluding servers cannot learn more than a fraction $\alpha$ of the file.
 In our model, unlike traditional secret sharing models \cite{b2,b3},  no
information-theoretically secure channels are available cost-free. Instead, the user communicates with each server through a one-way public channel, and shares with each server  a secret key, comprising a uniformly distributed sequence of $n$ bits. This formulation introduces a new communication–privacy trade-off that captures how communication resources, secret keys, and storage resources jointly determine the achievable storage rate. Our main contribution is the characterization of the optimal trade-off among the privacy leakage parameter $\alpha$, the secret-key length, and the file length that a user can store. Furthermore, for this optimal trade-off, we determine
the minimum amount of local randomness needed by the user, the minimum
amount of public communication, and the minimum amount
of storage needed at the servers.  In the special case $\alpha=0$, corresponding to perfect privacy (i.e., no privacy leakage), our results recover those of the prior work in \cite{b1}, which considered only the perfect-privacy setting. We also investigate an extension of our setting to a multilevel privacy constraint, where for each level $d \in \{1, \ldots, D\}$, any set of $z_d$ colluding servers cannot learn more than a fraction $\alpha_d$ of the~file, offering a finer privacy leakage control than the case $D=1$. 

\subsection{Comparison With Previous Work \cite{b1}}

Although the perfect-privacy case of our model reduces to the result of~\cite{b1} when $\alpha=0$, the general leakage setting considered here introduces new technical challenges. In~\cite{b1}, the converse arguments exploit the fact that any set of up to $z$ colluding servers learns nothing about the file. Here, however, such a coalition may already learn as much as $\alpha H(F)$, so the analysis must explicitly account for this nonzero leakage term and track its effect throughout the converse. Specifically, one must quantify how the allowed leakage combines with the public communication and the shared secret keys, and how it modifies the optimal tradeoff between file size, public communication, and local randomness. In addition, the achievability proof must show that leakage-allowing secret sharing can still be combined optimally with key-assisted delivery over public channels, beyond the perfect-privacy regime.

These challenges become even more pronounced in the multilevel extension, where for each level $d\in\{1,\dots,D\}$, every coalition of at most $z_d$ servers must satisfy a distinct leakage constraint $\alpha_d H(F)$. In particular, the converse must remain valid for all pairs $(z_d,\alpha_d)$ at once, while the achievability must construct a single strategy that meets all these constraints simultaneously. These considerations are absent from~\cite{b1}, which considered a single level, i.e., $D=1$, with perfect privacy, i.e., $\alpha =0$. 

Note that the multilevel privacy extension considered here is also conceptually different from the multiple-file setting  in~\cite{b1}. The multiple-file model increases the number of stored files, whereas our multilevel formulation still considers a single file $F$ but imposes several simultaneous privacy constraints indexed by the collusion level. Thus, the challenge is not to protect several files separately, but to design one scheme that guarantees for a single file multiple nested leakage constraints simultaneously. 

\subsection{Comparison With Other Related Work}

The most challenging aspect of our study is to establish converse results for the maximum file size that the user can store, the optimal amount of local randomness needed by the user, and the optimal amount of public communication between the user and servers.  Unlike conventional secret sharing models, including those that consider privacy leakage \cite{Morteza2,b23,yoshida2018optimal,yamamoto1986secret,blakley1984security}, our model requires  considering in our analysis shared secret keys, open communication between all parties, and a joint design of the share creation phase and the share distribution phase. Indeed, conventional secret sharing models treat these two phases separately, focusing only on the creation phase of the shares and assuming that the secure communication phase occurs through information-theoretically secure channels without any cost. Then, we establish an achievability result that matches the converse bound and show that separating share generation from share distribution is optimal. Specifically, shares are first generated using a secret-sharing scheme, and are then securely delivered to the servers using a one-time pad together with the shared secret keys. 

Note that a related line of work has investigated  secret sharing via public communication in settings where, instead of shared secret keys, the dealer and participants have access to noisy resources. These noisy resources are typically obtained either from a communication channel \cite{zou2015information,rana2023short,sultana2023secret,chou2022explicit,welling} or a source of correlated randomness   \cite{csiszar2010capacity,chou2018secret,chou2019biometric,rana2020secret,chou2021distributed,sultana2025secret,Miller}. Due to the presence of noise, these models require coding schemes and converse techniques that differ fundamentally from those employed in this paper.  Moreover, these two settings are themselves closely related to the problems of compound wiretap channels   \cite{liang2009compound} and compound secret-key generation  \cite{tavangaran2016secret,bloch2010channel, chou2013secret,yachongka2024secret,yachongka2025secret}, respectively.

\subsection{Paper Organization} Section \ref{sec2} provides a formal statement of the problem. Section \ref{sec3} presents our main results. Section \ref{sec2b} presents a multilevel privacy extension of the problem described in Section \ref{sec2}.  Finally,  Section \ref{sec4} provides concluding remarks.

\section{Problem Statement} \label{sec2}
Notation: For  $a \in \mathbb{N}^{*}$, define $\llbracket 1,a \rrbracket \triangleq [ 1,a] \cap \mathbb{N}$. For   $x \in \mathbb{R}$, define $[x]^+ \triangleq \max (0,x)$. For a set $\mathcal{S}$, let $2^{S}$ denote the power set of $\mathcal{S}$.  Let also $\bigtimes$ denote the Cartesian product.

Consider $L$ servers indexed by $\mathcal{L} \triangleq \llbracket 1,L \rrbracket$ and one user. Assume that Server $l \in \mathcal{L}$ and the  user share a secret key $K_{l} \in \mathcal{K} \triangleq \{0,1\}^n$, which is a sequence of $n$ bits uniformly distributed over $\{0,1\}^n$. The $L$ keys are assumed to be jointly independent. For any $\mathcal{Y}\subseteq \mathcal{L}$, we use the notation $K_{\mathcal{Y}} \triangleq(K_{y})_{y \in \mathcal{Y}}.\\$ 
\begin{definition}
A $(2^{r^{(F)}},2^{r^{(R)}},(2^{r^{(M)}_{l}})_{l \in \mathcal{L}}, (2^{r^{(S)}_{l}})_{l \in \mathcal{L}}) $ private storage strategy consists of  \label{def1}
\end{definition}
\begin{itemize}
    \item A file $F$ owned by the  user, which is uniformly distributed over $\mathcal{F}\triangleq \{0,1\}^{r^{(F)}}$ and independent from the keys $K_{\mathcal{L}}$;
    \item A sequence of local randomness $R$ owned by the  user, which is uniformly distributed over $\mathcal{R}\triangleq \{0,1\}^{r^{(R)}}$ and independent from all other random variables;
    \item $L$ encoding functions $h_l: \mathcal{R} \times \mathcal{K} \times \mathcal{F} \xrightarrow{} \mathcal{M}_{l}$, where $l \in \mathcal{L}$ and $\mathcal{M}_{l}\triangleq \{0,1\}^{r^{(M)}_{l}}$;
    \item $L$ servers with storage space $r^{(S)}_{l}$ bits for Server $l \in \mathcal{L}$;
    \item $L$ encoding functions $g_l: \mathcal{M}_{l} \times \mathcal{K} \xrightarrow{} \mathcal{S}_{l}$, where $l \in \mathcal{L}$ and $\mathcal{S}_{l}\triangleq \{0,1\}^{r^{(S)}_{l}}$;
    \item $2^L$ decoding function $f_{\mathcal{A}}: \bigtimes_{l \in \mathcal{A}} \mathcal{S}_{l} \xrightarrow{} \mathcal{F}$, where $\mathcal{A}\subseteq \mathcal{L}$;
    $\\$and operates as follows:
\end{itemize}
\begin{enumerate}
\item The  user publicly sends to Server $l \in \mathcal{L}$ the message $M_{l} \triangleq h_{l}(R, K_{l} , F)$. We use the notation $M \triangleq(M_{l})_{l \in \mathcal{L}}$.
\item Server $l \in \mathcal{L}$  stores $S_{l} \triangleq g_{l}(M_{l}, K_{l})$.
\item Any subset of servers $\mathcal{A} \subseteq \mathcal{L}$ can compute $\hat F(\mathcal{A}) \triangleq f_{\mathcal{A}}(S_{\mathcal{A}})$ an estimate of  $F$, where $S_{\mathcal{A}}\triangleq (S_{l})_{l \in \mathcal{A}}.\\$
\end{enumerate}
\begin{definition}\label{def2a} Fix $  t \in \llbracket 1,L \rrbracket, \phantom{-}  z \in  \llbracket 1, t-1 \rrbracket  \phantom{-} and \phantom{-}  \alpha \in [0,1]\cap \mathbb{Q}$. Then, $r^{(F)}$ is $(\alpha,t,z)$-achievable if there exists a
$(2^{r^{(F)}},2^{r^{(R)}},(2^{r^{(M)}_{l}})_{l \in \mathcal{L}}, (2^{r^{(S)}_{l}})_{l \in \mathcal{L}}) $
private file storage strategy such that
\begin{align}
  \max_{\forall \mathcal A \subseteq \mathcal{L}:|\mathcal{A}| \geq t } H(F|\hat F(\mathcal{A}))&=0,\phantom{-} \text{(Reliability)} \label{marker1}\\
\max_{\forall \mathcal {U}  \subseteq \mathcal{L}:|\mathcal{U}| \leq z }  \frac{I(F;M,K_{\mathcal{U}})}{H(F)} &\leq \alpha,\phantom{-} \text{(Privacy leakage)}  \label{marker2}\\
  I(F;M)&=0.\phantom{-} \text{(Security)} \label{marker3} 
\end{align}
The set of all $(\alpha,t,z)$-achievable lengths $r^{(F)}$ is denoted by $\mathcal{C}_{F}(\alpha,t,z)$. 
\end{definition}
\eqref{marker1} means that any subset of servers with size greater than or equal to $t$ is able to perfectly recover the file $F$,  \eqref{marker2} means that any subset of servers with size smaller than or equal to $z$ can learn at most a fraction $\alpha$ of the file, and \eqref{marker3} means that the public communication does not leak any information about the file. For a fixed $\alpha$, our main objective is to determine, under the constraints \eqref{marker1}-\eqref{marker3}, the maximal file length that the user can store in the servers given that the secret keys shared with the servers have length $n$. Next, another of our objectives is to determine (i) the minimum amount of local randomness at the user, (ii) the minimum storage requirement at the servers, and (iii) the minimum amount of public communication from the user to the servers that are needed to achieve the largest file rate in $\mathcal{C}_{F}(\alpha,t,z)$. To this end, we introduce the following definition. 
\begin{definition}\label{def2} Fix $\alpha \in [0,1]\cap \mathbb{Q}$, $t \in \llbracket 1,L\rrbracket $ and $ z \in \llbracket1, t-1\rrbracket$. For $r^{(F)}$ in $\mathcal{C}_{F}(\alpha,t,z)$, let $\mathcal{Q}(r^{(F)})$ be the set of tuples $T\triangleq({r^{(R)}},(r^{(M)}_{l})_{l \in \mathcal{L}}, (r^{(S)}_{l})_{l \in \mathcal{L}})$ such that there exists a $(2^{r^{(F)}},2^{r^{(R)}},(2^{r^{(M)}_{l}})_{l \in \mathcal{L}}, (2^{r^{(S)}_{l}})_{l \in \mathcal{L}}) $ private file storage
strategy that $(\alpha,t,z)$-achieves $r^{(F)}$. Then, define
\begin{align}
    &r^{(F)}_{*} (\alpha,t,z) \triangleq \sup_{r^{(F)} \in \mathcal{C}_{F} (\alpha,t,z)} r^{(F)},  \label{marker4}\\
    &r^{(M)}_{l,*} (\alpha,t,z) \triangleq \inf_{T \in \mathcal{Q} (r^{(F)}_{*}(\alpha,t,z))} r^{(M)}_{l},l \in \mathcal{L},\label{marker5}\\
    &r^{(M)}_{\sum,*} (\alpha,t,z) \triangleq \inf_{T \in \mathcal{Q} (r^{(F)}_{*}(\alpha,t,z))} \sum_{l \in \mathcal{L}}r^{(M)}_{l},\label{marker6}\\
    &r^{(R)}_{*} (\alpha,t,z) \triangleq \inf_{T \in \mathcal{Q} (r^{(F)}_{*}(\alpha,t,z))} r^{(R)},\label{marker7}\\
    &r^{(S)}_{l,*} (\alpha,t,z) \triangleq \inf_{T \in \mathcal{Q} (r^{(F)}_{*}(\alpha,t,z))} r^{(S)}_{l},l \in \mathcal{L} \label{marker8}.
\end{align}
\end{definition}
$r^{(F)}_{*}(\alpha,t,z)$ is the largest file size that the user can privately store under the constraints \eqref{marker1}-\eqref{marker3}. Then,  $r^{(R)}_{*} (\alpha,t,z)$,  $r^{(M)}_{l,*} (\alpha,t,z)$,  $r^{(M)}_{\sum,*} (\alpha,t,z) $, and $r^{(S)}_{*}(\alpha,t,z)$, $l \in \mathcal{L}$, correspond to the least amount of local randomness, the minimum amount of public communication to Server~$l$, the minimum amount of public communication to all the servers, and the minimum storage size required at Server~$l$, respectively, needed for the user to achieve $r^{(F)}_{*} (\alpha,t,z)$. 

Although Definition \ref{def2} may  appear stringent, given that each optimal quantity is defined independently of the others, our results demonstrate the existence of a storage strategy that simultaneously achieve all these optimal quantities.

\section{Main Results} \label{sec3}

In Section~\ref{secIIIA}, we first present our converse results. In Section~\ref{sec:res}, we then introduce an achievability scheme that matches the converse bounds. The achievability scheme and its analysis are presented in Section~\ref{secachiev6}. All achievability results assume that the key and secret lengths are chosen so that all displayed lengths are integers and the underlying secret-sharing constructions exist over a finite field of characteristic two. To streamline the presentation, the proofs of the converse results are deferred to the appendix.

\subsection{Impossibility Results} \label{secIIIA}
Let $\alpha \in [0,1]\cap \mathbb{Q}$, $t \in \llbracket1,L\rrbracket$ and $z \in \llbracket1, t-1\rrbracket$.
\begin{thm}
(Converse on the file length).  We have \label{thm1}
\begin{align*}
    r^{(F)}_{*}(\alpha,t,z) \leq \min\left(\frac{t-z}{1-\alpha}, t\right)n.
\end{align*}\end{thm}
\begin{proof} Set $D=1$ in Appendix \ref{appendix:a}.\end{proof}
Theorem~\ref{thm1} means that it is impossible for the user to store a file of length larger than $\min\left(\frac{t-z}{1-\alpha}, t\right)n$ bits.
\begin{thm} (Converse on storage size requirement at the server). We have \label{thm2}
\begin{align*}
   & r^{(S)}_{l,*}(\alpha,t,z)\geq n, \forall  l \in \mathcal{L}.
\end{align*}\end{thm}
\begin{proof} Set $D=1$ in Appendix  \ref{appendix:b}.\end{proof}
Theorem~\ref{thm2} means that Server~$l \in \mathcal{L}$ needs a storage capacity of at least $n$ bits.
\begin{thm}(Converse on the total amount of public communication to the servers). We have \label{thm3}
\begin{align} 
    &\max \left( \frac{1-\alpha}{t-z},  \frac{1}{t}\right)L \times r_{*}^{(F)}(\alpha,t,z)\leq r_{\sum,*}^{(M)}(\alpha,t,z).\label{marker9}
\end{align}
\end{thm}
\begin{proof} Set $D=1$ in Appendix \ref{appendix:c}.\end{proof}
Theorem~\ref{thm3} means that it is impossible for the user to store a file of length $r^{(F)}_{*}(\alpha,t,z)$ if the public communication sum length to the servers is smaller than $\max \left( \frac{1-\alpha}{t-z},  \frac{1}{t}\right)L \times r_{*}^{(F)}(\alpha,t,z)$ bits.
\begin{thm}(Converse on the amount of public communication
to an individual server).  
Consider the following condition
\begin{align}
   \forall \mathcal{U}, \mathcal{V} \subseteq \mathcal{L}, |\mathcal{U}| =|\mathcal{V}| \Rightarrow I(F;M_{\mathcal{U}},K_{\mathcal{U}})=I(F;M_{\mathcal{V}},K_{\mathcal{V}}). \label{marker10}
\end{align}
\eqref{marker10} indicates that any two sets of colluding servers that have the same size have the same amount of information about the file $F$. If \eqref{marker10} holds, then we have
\begin{align*}
   \max\left(\frac{1-\alpha}{t-z} , \frac{1}{t} \right) \times r^{(F)}_* (\alpha,t,z)\leq r_{l,*}^{(M)} (\alpha,t,z),\forall l \in \mathcal{L}. 
\end{align*} \label{thm4}
\end{thm}
\begin{proof} Set $D=1$ in Appendix \ref{appendix:d}.\end{proof}
Under the condition \eqref{marker10}, Theorem~\ref{thm4} means that it is impossible for the user to store a file of length $ r_{*}^{(F)}$ if the public communication length to Server $l \in \mathcal{L}$ is smaller than $ \max\left(\frac{1-\alpha}{t-z}, \frac{1}{t} \right) \times r^{(F)}_* (\alpha,t,z)$~bits.
\begin{thm}\label{thm5}
 (Converse on the amount of required local randomness at the users). We have
\begin{align}
  r^{(R)}_{*}(\alpha,t,z) \geq \frac{[z-\alpha t]^+}{t-z}r^{(F)}_{*} (\alpha,t,z) \label{marker11}.
\end{align}
\end{thm}

\begin{proof}
See Appendix~\ref{appendix:e}.
\end{proof}Theorem~\ref{thm5} means that it is impossible for the  user to store a file of length $r^{(F)}_{*}(\alpha,t,z)$ if the amount of its local randomness is smaller than $\frac{[z-\alpha t]^+}{t-z}r^{(F)}_{*} (\alpha,t,z)$ bits.

\subsection{Capacity Result} \label{sec:res}
 Let $\alpha \in [0,1]\cap \mathbb{Q}$, $ t \in \llbracket 1, L\rrbracket$ and $ z\in \llbracket 1, t-1\rrbracket$. In the following, we propose an achievability scheme that simultaneously matches all the converse results of Section~\ref{secIIIA}.
\begin{thm}
    There exists a $(2^{r^{(F)}},2^{r^{(R)}},(2^{r^{(M)}_{l}})_{l \in \mathcal{L}}, (2^{r^{(S)}_{l}})_{l \in \mathcal{L}}) $ private file storage strategy that $(\alpha, t, z)$-achieves $r^{(F)}$such that, 
    \begin{align*}
        r^{(F)}&= \min \left(\frac{t-z}{1-\alpha},t \right)n,\\
r^{(R)}&=\frac{[z-\alpha t]^+}{1-\alpha} n,\\
r^{(S)}_{l}&= n, \forall l \in  \mathcal{L},\\
r^{(M)}_{l}&= n, \forall l \in  \mathcal{L}.
    \end{align*}
    \label{thm6}
\end{thm}
Theorem~\ref{thm6} establishes the optimality of separating
share generation from share distribution for the maximum
file length and, at this maximum, for the minimum local
randomness and total public communication. Indeed, the
converse bounds in Section~\ref{secIIIA}
are attained simultaneously by first generating
$(\alpha,t,z)$-secret shares $(H_l)_{l\in\mathcal{L}}$
from $(F,R)$, independently of the keys, and then
publicly sending $M_l=H_l\oplus K_l$ to Server~$l$.
Thus, within this model, jointly designing share
generation and share distribution procures no advantage. 
Theorem~\ref{thm6} also recovers the perfect-privacy
result in~\cite{b1} when $\alpha=0$.

From Theorem~\ref{thm6}, we observe two distinct regimes for $\alpha$. When $\alpha \geq \frac{z}{t}$ (respectively $\alpha < \frac{z}{t}$), the  user can store a file of size
at most $ nt$  bits (respectively $ \frac{t-z}{1-\alpha}n$ bits) such that any set of servers larger than
or equal to $t$ can reconstruct the file, and any set of servers
smaller than or equal to $z$ can learn at most a fraction
$\alpha$ of the file. Moreover, if the  user stores a file of length $ \frac{t-z}{1-\alpha} n$ bits (respectively $ nt$ bits), then the optimal storage capacity at each server is $ n$ bits,
the minimum amount of local randomness needed by the  user is
$\frac{z-\alpha t}{1-\alpha} n$ bits (respectively $0$ bits), and the optimal amount of public communication
from the user to all the servers is $L \times n $~bits.

\subsection{Proof of Theorem~\ref{thm6}} \label{secachiev6}

We first review the notion of secret sharing with privacy leakage before we present our coding scheme and its analysis.

\subsubsection{Review of secret sharing with privacy leakage} \label{secprelim1}
Let $\alpha \in [0,1]\cap \mathbb{Q}$, $t \in \llbracket1,L\rrbracket$ and $z\in \llbracket 1 , t-1 \rrbracket$.
\begin{definition}  \label{defssh}
An $(\alpha, t, z)$- secret sharing scheme consists of
\begin{itemize}
    \item  A secret $S$ uniformly distributed over $\{0, 1\}^{n_{s}};$
    \item  An encoder $ e : \{0, 1\}
^{n_{s}} \times \{0, 1\}^{n_{r}} \rightarrow \{0, 1\}^{n_{sh} L},(S, R)\mapsto {( H_{l})_{l\in \mathcal{L}}}$, which takes as input the secret $S$ and a randomization sequence $R$ uniformly distributed over $\{0, 1\}^{n_{r}}$ and independent of $S$, and outputs $L$ shares ${( H_{l})_{l\in \mathcal{L}}}$ of length $n_{sh}$. For any $\mathcal{S} \subseteq \mathcal{L}$, we define $H_{\mathcal{S}}={( H_{l})_{l\in \mathcal{S}}}; $ \end{itemize}
and satisfies the two conditions
\begin{align}
    \displaystyle \max_{\mathcal{T} \subseteq \mathcal{L}:|\mathcal{T}|=t} H(S|H_{\mathcal{T}})&=0, \phantom{-}  \text{(Recoverability)} \\
   \displaystyle \max_{\mathcal{U} \subseteq \mathcal{L}:|\mathcal{U}|\leq z } I(S;H_{\mathcal{U}})&\leq \alpha H(S). \phantom{-}\text{(Privacy leakage)} \label{eqpriv}
\end{align}
\end{definition}

\begin{thm}
(\cite{b23})  For a fixed secret length $n_{s}$, there exists an $(\alpha, t, z)$-secret sharing scheme such that the length of a share $n_{sh}$ and the length of the randomization sequence $n_{r}$ satisfy
 \begin{align*}
      n_{sh}=\max\left(\frac{1-\alpha}{t-z}, \frac{1}{t}\right) n_{s}, \quad
      n_{r}=\frac{[z-\alpha t]^+}{t-z} n_{s}.
 \end{align*}
\label{thm7}
\end{thm}
\subsubsection{Coding scheme and analysis}   $\\$

\textit{Coding scheme}: Consider a file $F$ such that
$r^{(F )} = |F| = \min\left( \frac{t-z}{1-\alpha},t\right)n$. Then, the  user forms $(H_{l})_{l\in \mathcal{L}}$
with an $(\alpha, t, z)$ secret sharing scheme  from
Theorem~\ref{thm7} applied to $F$. By Theorem~\ref{thm7}, for $l \in \mathcal{L}$, the length of a share is $|H_{l}| = \max (\frac{1-\alpha}{t-z}, \frac{1}{t})\times \min\left( \frac{t-z}{1-\alpha},t\right)n = n$, and the length of
the randomization sequence is $n_{r} =\frac{[z-\alpha t]^+}{t-z} \times \min\left( \frac{t-z}{1-\alpha},t\right)n $. Hence,
since $|K_{l}| = n$, $l \in \mathcal{L}$, the  user can form $M_{l}\triangleq H_{l} \oplus K_{l}$ and publicly send it to Server $l$, where $\oplus$ denotes bitwise
modulo-two addition. Upon receiving $M_{l}$, Server $l$ stores
$S_{l} \triangleq K_{l} \oplus M_{l} = H_{l}$.

\textit{Resources needed:} For  $l \in \mathcal{L}$, the length of the randomization sequence at the user is $r^{(R)} = n_{r} = \frac{[z-\alpha t]^+}{1-\alpha}n$, the length of the public communication from the user to Server~$l$ is $ r^{(M)}_{l} = |M_{l}|= n$, and the storage needed at Server~$l$ is $r^{(S)}_{l}=
| M_{l}|= n.$

 \textit{Analysis of recoverability:} Consider an arbitrary subset $\mathcal{A} \subseteq \mathcal{L}$ of $t$ servers that pool their information, they
then have access to $(S_{l})_{l \in \mathcal{A}} = (H_{l})_{l \in \mathcal{A}}$ such that by
Theorem~\ref{thm7}, $H(F|(H_{l})_{l \in \mathcal{A}}) = 0$ since $|\mathcal{A}|= t.$

\textit{Analysis of Security:}
We have
\begin{align*}
    I(F;M)&\stackrel{(a)}{=}I(F; H_{\mathcal{L}} \oplus K_{\mathcal{L}})\\
    & =H(H_{\mathcal{L}} \oplus K_{\mathcal{L}})-H(H_{\mathcal{L}} \oplus K_{\mathcal{L}}|F)\\
    &\leq |K_{\mathcal{L}}|-H(H_{\mathcal{L}} \oplus K_{\mathcal{L}}|F)\\
        &\stackrel{(b)}{\leq} |K_{\mathcal{L}}|-H(H_{\mathcal{L}} \oplus K_{\mathcal{L}}|F,H_{\mathcal{L}})\\
     &= |K_{\mathcal{L}}|-H( K_{\mathcal{L}}|F,H_{\mathcal{L}})\\
     &\stackrel{(c)}{=} |K_{\mathcal{L}}|-H(K_{\mathcal{L}})\\
     &\stackrel{(d)}{=}0,
\end{align*}
where \begin{enumerate}[(a)]
    \item 
holds because $M= H_{\mathcal{L}} \oplus K_{\mathcal{L}}$;
\item  holds   because conditioning reduces entropy;
\item holds   because the keys $K_{\mathcal{L}}$ are  independent of $(F,H_{\mathcal{L}})$; \item holds by  uniformity of $K_{\mathcal{L}}$.
\end{enumerate}

\textit{Analysis of privacy leakage:}  
Then, consider an arbitrary subset
$\mathcal{U} \subseteq \mathcal{L}$ of $z$ colluding servers. Then, we have
\begin{align*}
    I(F;M,K_{\mathcal{U}}) 
    & \stackrel{(a)}{=} I(F;H_{\mathcal{U}},M_{\mathcal{U}^c},K_{\mathcal{U}})\\
    & \stackrel{(b)}{=} I(F;H_{\mathcal{U}})+I(F;M_{\mathcal{U}^c},K_{\mathcal{U}}|H_{\mathcal{U}})\\
     & \stackrel{(c)}{\leq} \alpha H(F)+ I(F;M_{\mathcal{U}^c},K_{\mathcal{U}}|H_{\mathcal{U}})\\
     & \stackrel{(d)}{=}\alpha H(F)+I(F;K_{\mathcal{U}}|H_{\mathcal{U}})+I(F;M_{\mathcal{U}^c}|H_{\mathcal{U}},K_{\mathcal{U}})\\
     & \stackrel{(e)}{=} \alpha H(F)+ I(F;M_{\mathcal{U}^c}|H_{\mathcal{U}},K_{\mathcal{U}})\\
      & \leq \alpha H(F)+ I(F,H_{\mathcal{U}^c};M_{\mathcal{U}^c}|H_{\mathcal{U}},K_{\mathcal{U}})\\
      & \stackrel{}{=} \alpha H(F)+ H(M_{\mathcal{U}^c}|H_{\mathcal{U}},K_{\mathcal{U}})-H(M_{\mathcal{U}^c}|F,H_{\mathcal{U}^c},H_{\mathcal{U}},K_{\mathcal{U}})\\
      & \leq  \alpha H(F)+ |M_{\mathcal{U}^c}|-H(M_{\mathcal{U}^c}|F,H_{\mathcal{U}^c},H_{\mathcal{U}},K_{\mathcal{U}})\\
      & \stackrel{(f)}{=} \alpha H(F)+ |M_{\mathcal{U}^c}|-H(K_{\mathcal{U}^c}|F,H_{\mathcal{U}^c},H_{\mathcal{U}},K_{\mathcal{U}})\\
      & \stackrel{(g)}{=} \alpha H(F)+|M_{\mathcal{U}^c}|-H(K_{\mathcal{U}^c})\\
      & \stackrel{(h)}{=} \alpha H(F),
\end{align*}
where \begin{enumerate}[(a)]
    \item  holds because $M=(M_{\mathcal{U}^c},(H_{l}\oplus K_{l})_{l \in \mathcal{U}})$; \item and (d) hold by the chain rule; 
   \item holds by Theorem~\ref{thm7}; \addtocounter{enumi}{1} 
    \item holds because $(F,H_{\mathcal{U}})$ and $K_{\mathcal{U}}$ are independent;
    \item holds because $M_{\mathcal{U}^c} = (H_{\mathcal{U}^c} \oplus K_{\mathcal{U}^c})$; \item holds because the keys are jointly independent and independent of $(F, H_{\mathcal{L}})$; \item holds by uniformity of $K_{\mathcal{U}^c}$ and  $|K_{\mathcal{U}^c}|=|M_{\mathcal{U}^c}|$.
    \end{enumerate}

\section{Extension to a multilevel privacy constraint} \label{sec2b}
In this section, we extend the problem formulation of Section~\ref{sec2} to a multilevel privacy constraint setting. Specifically, Definition \ref{def1} remains unchanged and Definition \ref{def2a} is replaced as follows.

\begin{definition} Fix $  t \in \llbracket 1,L \rrbracket$, $\boldsymbol{z}\triangleq \left(z_{d}\right)_{d \in \llbracket 1, D\rrbracket}   \in  \llbracket 1, t-1 \rrbracket^{D}  $, $ \boldsymbol{\alpha}\triangleq \left(\alpha_{d}\right)_{d \in \llbracket 1, D\rrbracket} \in ~[0,1]^D\cap~\mathbb{Q}^D $. Then, $r^{(F)}$ is $(\boldsymbol{\alpha}, t, \boldsymbol{z})$-achievable if there exists a
$(2^{r^{(F)}},2^{r^{(R)}},(2^{r^{(M)}_{l}})_{l \in \mathcal{L}}, (2^{r^{(S)}_{l}})_{l \in \mathcal{L}}) $
private file storage strategy such that
\begin{align}
  \max_{\forall \mathcal A \subseteq \mathcal{L}:|\mathcal{A}| \geq t } H(F|\hat F(\mathcal{A}))&=0,\phantom{-} \label{marker1b}\\
\forall d \in \llbracket 1, D\rrbracket, \max_{\forall \mathcal {U}  \subseteq \mathcal{L}:|\mathcal{U}| \leq z_{d} }  \frac{I(F;M,K_{\mathcal{U}})}{H(F)} &\leq \alpha_{d},\phantom{-}  \label{marker2b}\\
  I(F;M)&=0.\phantom{-} \label{marker3b} 
\end{align}
\end{definition}
Note that the setting of Section \ref{sec2} is recovered by setting $D=1$ in \eqref{marker2b}. The constraint \eqref{marker2b} means that for any $d \in \llbracket 1, D\rrbracket$, any set of $z_d$ colluding servers cannot learn more than a fraction $\alpha_d$ of the file. Such a constraint offers a finer control of information leakage compared to the case $D=1$ treated in Section \ref{sec2}.

Similar to Definition \ref{def2}, we define $r^{(F)}_{*}(\boldsymbol{\alpha}, t, \boldsymbol{z})$ as the largest file size that the user can privately store under the constraints \eqref{marker1b}-\eqref{marker3b}. Then,    $r^{(M)}_{l,*} (\boldsymbol{\alpha}, t, \boldsymbol{z})$,  $r^{(M)}_{\sum,*} (\boldsymbol{\alpha}, t, \boldsymbol{z}) $, and $r^{(S)}_{*}(\boldsymbol{\alpha}, t, \boldsymbol{z})$, $l \in \mathcal{L}$, as  the minimum amount of public communication to Server~$l$, the minimum amount of public communication to all the servers, and the minimum storage size required at Server~$l$, respectively, needed for the user to achieve $r^{(F)}_{*} (\boldsymbol{\alpha}, t, \boldsymbol{z})$. 

\subsection{Impossibility Results} \label{ssecconv}
Let $  t \in \llbracket 1,L \rrbracket$, $\boldsymbol{z}\triangleq \left(z_{d}\right)_{d \in \llbracket 1, D\rrbracket}   \in  \llbracket 1, t-1 \rrbracket^{D}  $, $ \boldsymbol{\alpha}\triangleq \left(\alpha_{d}\right)_{d \in \llbracket 1, D\rrbracket} \in ~[0,1]^D\cap~\mathbb{Q}^D $. For $d=0$, define also $\alpha_{0} \triangleq 0$ and $z_{0} \triangleq 0$.
\begin{thm}
(Converse on the file length). We have \label{thm1b}
\begin{align*}
    r^{(F)}_{*}(\boldsymbol{\alpha}, t, \boldsymbol{z}) \leq \min_{d \in \llbracket 0, D\rrbracket}\left(\frac{t-z_{d}}{1-\alpha_{d}}\right)n.
\end{align*}\end{thm}
\begin{proof} See Appendix \ref{appendix:a}.\end{proof}
Theorem~\ref{thm1b} means that it is impossible for the user to store a file of length larger than $\displaystyle\min_{d \in \llbracket 0, D\rrbracket}\left(\frac{t-z_{d}}{1-\alpha_{d}}, t\right)n$ bits.
\begin{thm} (Converse on storage size requirement at the server). We have \label{thm2b}
\begin{align*}
   & r^{(S)}_{l,*}(\boldsymbol{\alpha}, t, \boldsymbol{z})\geq n, \forall  l \in \mathcal{L}.
\end{align*}\end{thm}
\begin{proof} See Appendix  \ref{appendix:b}.\end{proof}
Theorem~\ref{thm2b} means that  Server~$l \in \mathcal{L}$ needs a storage capacity of at least $n$ bits.
\begin{thm}(Converse on the total amount of public communication to the servers). We have \label{thm3b}
\begin{align} 
    &\max_{d \in \llbracket 0, D\rrbracket} \left( \frac{1-\alpha_{d}}{t-z_{d}}\right)L \times r_{*}^{(F)}(\boldsymbol{\alpha}, t, \boldsymbol{z})\leq r_{\sum,*}^{(M)}(\boldsymbol{\alpha}, t, \boldsymbol{z}).\label{marker9b}
\end{align}
\end{thm}
\begin{proof} See Appendix \ref{appendix:c}.\end{proof}
Theorem~\ref{thm3b} means that it is impossible for the user to store a file of length $r^{(F)}_{*}(\boldsymbol{\alpha}, t, \boldsymbol{z})$ if the public communication sum length to the servers is smaller than $\displaystyle\max_{d \in \llbracket 0, D\rrbracket} \left( \frac{1-\alpha_{d}}{t-z_{d}}\right)L \times r_{*}^{(F)}(\boldsymbol{\alpha}, t, \boldsymbol{z})$ bits.

\begin{thm}(Converse on the amount of public communication
to an individual server).  
 If the leakage symmetry condition~\eqref{marker10} holds, then we have
\begin{align*}
   \max_{d \in \llbracket 0, D\rrbracket}\left(\frac{1-\alpha_{d}}{t-z_{d}}\right) \times r^{(F)}_* (\boldsymbol{\alpha}, t, \boldsymbol{z})\leq r_{l,*}^{(M)} (\boldsymbol{\alpha}, t, \boldsymbol{z}),\forall l \in \mathcal{L}. 
\end{align*} \label{thm4b}
\end{thm}
\begin{proof} See Appendix \ref{appendix:d}.\end{proof}
 Under condition \eqref{marker10}, Theorem~\ref{thm4b} means that it is impossible for the user to store a file of length $ r_{*}^{(F)}$ if the public communication length to Server $l \in \mathcal{L}$ is smaller than $ \displaystyle\max_{d \in \llbracket 0, D\rrbracket}\left(\frac{1-\alpha_{d}}{t-z_{d}}\right) \times r^{(F)}_* (\boldsymbol{\alpha}, t, \boldsymbol{z})$~bits.
 \subsection{Capacity Result}
 In the following, we propose an achievability scheme that simultaneously matches all the converse results of Section~\ref{ssecconv}.
\begin{thm}
   There exists a $(2^{r^{(F)}},2^{r^{(R)}},(2^{r^{(M)}_{l}})_{l \in \mathcal{L}}, (2^{r^{(S)}_{l}})_{l \in \mathcal{L}}) $ private file storage strategy that $(\boldsymbol{\alpha}, t, \boldsymbol{z})$-achieves $r^{(F)}$such that 
    \begin{align*}
         r^{(F)} &= \min_{d \in \llbracket 0, D\rrbracket}\left(\frac{t-z_{d}}{1-\alpha_{d}}\right)n,\\
r^{(S)}_{l}&= n, \forall l \in  \mathcal{L},\\
r^{(M)}_{l}&= n, \forall l \in  \mathcal{L}.
    \end{align*}
    \label{thm6b}
\end{thm}
\subsection{Proof of Theorem~\ref{thm6b}} 
In Section \ref{secprelim}, we first present a generalization of secret sharing with privacy leakage, extending the model presented in Section \ref{secprelim1}. In Section \ref{seccs2}, we then describe our achievability scheme and provide its analysis. 

\subsubsection{Preliminaries} \label{secprelim} Let $  t \in \llbracket 1,L \rrbracket$, $\boldsymbol{z}\triangleq \left(z_{d}\right)_{d \in \llbracket 1, D\rrbracket}   \in  \llbracket 1, t-1 \rrbracket^{D}  $, $ \boldsymbol{\alpha}\triangleq \left(\alpha_{d}\right)_{d \in \llbracket 1, D\rrbracket} \in ~[0,1]^D\cap~\mathbb{Q}^D $. Definition~\ref{defssh} can be extended by replacing \eqref{eqpriv} with the constraint 
\begin{align}
   \displaystyle \max_{\mathcal{U} \subseteq \mathcal{L}:|\mathcal{U}|\leq z_{d} } I(S;H_{\mathcal{U}})\leq \alpha_{d} H(S). \label{marker13b}
\end{align}
This new definition defines an $(\boldsymbol{\alpha}, t, \boldsymbol{z})$- secret sharing scheme. The following theorem characterizes the optimal share size for such an $(\boldsymbol{\alpha}, t, \boldsymbol{z})$- secret sharing scheme.

\begin{thm}
[\cite{Morteza}] For a fixed secret length $n_{s}$, there exists an $(\boldsymbol{\alpha}, t, \boldsymbol{z})$-secret sharing scheme such that the length of a share $n_{sh}$  satisfies
 \begin{align*}
n_{sh}&=\displaystyle\max_{d \in \llbracket 0, D\rrbracket}\left(\frac{1-\alpha_{d}}{t-z_{d}}\right) n_{s}. \end{align*}
\label{thm7b}
\end{thm}

\subsubsection{Coding scheme and analysis} \label{seccs2} Consider a file $F$ such that
$r^{(F )} ~= ~|F|~ =~\displaystyle \min_{d \in \llbracket 0, D\rrbracket}\left(\frac{t-z_{d}}{1-\alpha_{d}}\right)n$. Then, the  user forms $(H_{l})_{l\in \mathcal{L}}$
with an $(\boldsymbol{\alpha}, t, \boldsymbol{z})$ secret sharing scheme from
Theorem~\ref{thm7b} such that, for $l \in \mathcal{L}$, the length of a share is $|H_{l}| = ~\displaystyle\max_{d \in \llbracket 0, D\rrbracket}\left(\frac{1-\alpha_{d}}{t-z_{d}}\right)\times \displaystyle \min_{d \in \llbracket 0, D\rrbracket}\left(\frac{t-z_{d}}{1-\alpha_{d}}\right)n = n$. Hence,
since $|K_{l}| = n, l \in \mathcal{L}$, the  user can form $M_{l}\triangleq H_{l} \oplus K_{l}$ and publicly send it to Server $l$, where $\oplus$ denotes bitwise
modulo-two addition. Upon receiving $M_{l}$, Server $l$ stores
$S_{l} \triangleq K_{l} \oplus M_{l} = H_{l}$.

\textit{Resources needed:} For  $l \in \mathcal{L}$, the length of the public communication from the user to Server~$l$ is $ r^{(M)}_{l} = |M_{l}|= n$, and the storage needed at Server~$l$ is $r^{(S)}_{l}=
| M_{l}|= n.$

 \textit{Analysis of recoverability:} Consider an arbitrary subset $\mathcal{A} \subseteq \mathcal{L}$ of $t$ servers that pool their information, they
then have access to $(S_{l})_{l \in \mathcal{A}} = (H_{l})_{l \in \mathcal{A}}$ such that by
Theorem~\ref{thm7b}, $H(F|(H_{l})_{l \in \mathcal{A}}) = 0$ since $|\mathcal{A}|= t.$

\textit{Analysis of privacy leakage:}  Let $d \in \llbracket 1, D \rrbracket$ and consider an arbitrary subset
$\mathcal{U} \subseteq \mathcal{L}$ of $z_{d}$ colluding servers. Then, similar to Section~\ref{secachiev6}, we have
\begin{align*}
I(F;M,K_{\mathcal{U}}) \leq \alpha_{d} H(F).
\end{align*}

The analysis of security is identical to the one in Section~\ref{secachiev6}.
\section{Concluding remarks}\label{sec4}
We considered storing a file across $L$ servers, ensuring that any set of at least $t$ servers can reconstruct the file and any set of $z$ colluding servers cannot learn more than a fraction $\alpha $  of the file. Hence, our model introduces a trade-off between the privacy leakage parameter $\alpha $  and the file length that the user can store on the servers. In contrast to traditional secret sharing, our model did not assume the existence of information-theoretically secure channels. Instead, the user communicates with the servers over a public channel and shares a secret key of length $n$ with each server. For a given $n$, we characterized the optimal trade-off between the file length that the user can store and the privacy leakage parameter $\alpha $. Additionally, when this optimal trade-off is achieved, we established the minimum amount of local randomness needed by the user,  the minimum amount of publication communication, and the minimum storage requirement at the servers. We also extended our setting with a multilevel privacy constraint, where for each level $d \in \llbracket 1, D\rrbracket$, any set of $z_d$ colluding servers cannot learn more than a fraction $\alpha_d$ of the~file. For this multilevel privacy constraint setting, an open problem is to characterize the optimal amount of local randomness needed at the user.

\appendices
 \section{Proof of Theorem \ref{thm1b}}\label{appendix:a}
  Consider an arbitrary $(2^{r^{(F)}},2^{r^{(R)}},(2^{r^{(M)}_{l}})_{l \in \mathcal{L}}, (2^{r^{(S)}_{l}})_{l \in \mathcal{L}}) $ private file storage  strategy that $(\boldsymbol{\alpha}, t, \boldsymbol{z})$-achieves $r^{(F)}.\\$

Let $ d\in \llbracket0,D\rrbracket$, $\mathcal{A},\mathcal{U}\subseteq \mathcal{L}$ such that $|\mathcal{A}|=t$, $|\mathcal{U}|=z_{d}$, and $\mathcal{U}\subset \mathcal{A}$.
      We have
\begin{align}
   r^{(F)} \nonumber \nonumber&\stackrel{(a)}{=} H(F)\\
  \nonumber &\stackrel{}{=}H(F|M,K_{\mathcal{U}})+I(F;M,K_{\mathcal{U}}) \\
 \nonumber & \stackrel{(b)}{\leq} H(F|M,K_{\mathcal{U}}) +  \alpha_{d}  H(F)  \\
   \nonumber&\stackrel{}{=} I(\hat F(\mathcal{A}); F|M,K_{\mathcal{U}}) + H(F|M,K_{\mathcal{U}}, \hat F(\mathcal{A}))  + \alpha_{d}  H(F)  \\
   \nonumber&\stackrel{(c)}{\leq}  I(\hat F(\mathcal{A}); F|M,K_{\mathcal{U}}) + H(F|\hat F(\mathcal{A})) + \alpha_{d}  H(F)\\
 \nonumber & \stackrel{(d)}{=}I(\hat F(\mathcal{A}); F|M,K_{\mathcal{U}}) + \alpha_{d}  H(F)\\
   \nonumber&\stackrel{(e)}{\leq}I(M,K_{\mathcal{A}}; F|M,K_{\mathcal{U}}) + \alpha_{d}  H(F)\\
   \nonumber&\stackrel{(f)}{=}I(K_{\mathcal{A}}; F|M,K_{\mathcal{U}}) + \alpha_{d}  H(F)\\
  \nonumber& \stackrel{(g)}{\leq} I(K_{\mathcal{A}};K_{\mathcal{L}}, F,R|K_{\mathcal{U}})+ \alpha_{d}  H(F)\\
   \nonumber&=I(K_{\mathcal{A}};K_{\mathcal{L}}, F|K_{\mathcal{U}}) +I(K_{\mathcal{A}};R |K_{\mathcal{L}},F) + \alpha_{d}  H(F)\\
 \nonumber & \stackrel{}{=}I(K_{\mathcal{A}};K_{\mathcal{L}}, F|K_{\mathcal{U}})  + \alpha_{d}  H(F)\\
 \nonumber & {=}I(K_{\mathcal{A}};K_{\mathcal{L}}|K_{\mathcal{U}})  + \alpha_{d}  H(F)\\
  \nonumber &\stackrel{}{=}H(K_{\mathcal{A}}| K_{\mathcal{U}})+ \alpha_{d}  H(F)\\
 \nonumber & \stackrel{(h)}{=}H(K_{\mathcal{A\backslash U}})+ \alpha_{d}  H(F)\\
   \nonumber&\stackrel{(i)}{=}n(t-z_{d})+\alpha_{d}  H(F)\\
   &=n(t-z_{d})+\alpha_{d}  r^{(F)},\label{marker14}
\end{align}
    where \begin{enumerate}[(a)]
        \item  holds by uniformity of the file $F$; \item holds by~\eqref{marker2b} for $d\in \llbracket1,D\rrbracket$, otherwise holds by~\eqref{marker3b} for $d=0$ because in this case $\alpha_{0} \triangleq 0$ and $z_{0} \triangleq 0$; \item  holds because conditioning reduces entropy; \item holds by~\eqref{marker1b}; \item holds by the data processing inequality; \item holds by the chain rule; \item holds by the chain rule and the data processing inequality because $M$ is a function of $(F, R,K_{\mathcal{L}})$; \item holds because $\mathcal{U} \subset \mathcal{A}$; \item holds because the keys are distributed uniformly and $|\mathcal{A}\backslash \mathcal{U}|=t-z_{d}$.\end{enumerate} Then, from \eqref{marker14} we have
\begin{align} 
   &\forall d \in \llbracket 0, D\rrbracket,  r^{(F)} \leq \frac{n(t-z_{d})}{1-\alpha_{d}}. \label{marker14c}
\end{align}

Finally, note that \eqref{marker14c} is valid for any private file storage strategy and, in particular, the ones that $(\boldsymbol{\alpha}, t, \boldsymbol{z})$-achieve $r^F_{*}(\boldsymbol{\alpha}, t, \boldsymbol{z}).$
\section{Proof of Theorem \ref{thm2b}}\label{appendix:b}Consider an arbitrary $(2^{r^{(F)}},2^{r^{(R)}},(2^{r^{(M)}_{l}})_{l \in \mathcal{L}}, (2^{r^{(S)}_{l}})_{l \in \mathcal{L}}) $ private file storage  strategy that $(\boldsymbol{\alpha}, t, \boldsymbol{z})$-achieves $r^{(F)}.$
    
  Server $l \in \mathcal{L}$ must store the key ${K}_{l}$ at the beginning of the protocol. Hence, we must have 
\begin{align*}
   & r^{(S)}_{l,*}(\boldsymbol{\alpha}, t, \boldsymbol{z})\geq |{K}_{l}|= n. 
\end{align*}
   \section{Proof of Theorem \ref{thm3b}}\label{appendix:c}
  Consider an arbitrary $(2^{r^{(F)}},2^{r^{(R)}},(2^{r^{(M)}_{l}})_{l \in \mathcal{L}}, (2^{r^{(S)}_{l}})_{l \in \mathcal{L}})$ private file storage strategy that $(\boldsymbol{\alpha}, t, \boldsymbol{z})$-achieves $r^{(F)}.$

  We will use the following lemma, which provides a lower bound on the amount of communication from the user to subsets of servers.
\begin{lem} For $ d\in \llbracket0,D\rrbracket$, $\mathcal{T} \subseteq \mathcal{L}$ and $\mathcal{S} \subseteq \mathcal{L} \backslash \mathcal{T}$ such that $|\mathcal{T}|=z_{d}$ and $|\mathcal{S}|=t-z_{d}$, we have \label{lem2}
\begin{align}
 \displaystyle \sum_{l \in S} H(M_{l}) \geq (1-\alpha_d)H(F). \label{marker16} 
\end{align}\end{lem}
\begin{proof}
     We have
\begin{align*}
\displaystyle\sum_{l \in \mathcal{S}} H(M_{l})+\sum_{l \in \mathcal{S}} H(K_{l})
 &\geq H(M_{\mathcal{S}}, K_{\mathcal{S}})\\
 &\stackrel{(a)}{\geq}H(M_{\mathcal{S}}, K_{\mathcal{S}}| M_{\mathcal{T}}, K_{\mathcal{T}})\\
 &\stackrel{}{=} I(M_{\mathcal{S}},K_{\mathcal{S}};F|M_{\mathcal{T}}, K_{\mathcal{T}})+H(M_{\mathcal{S}},K_{\mathcal{S}}|F,M_{\mathcal{T}}, K_{\mathcal{T}})\\
  &\stackrel{}{=}H(F|M_{\mathcal{T}}, K_{\mathcal{T}})-H(F|M_{\mathcal{S}},K_{\mathcal{S}},M_{\mathcal{T}}, K_{\mathcal{T}})+H(M_{\mathcal{S}},K_{\mathcal{S}}|F,M_{\mathcal{T}}, K_{\mathcal{T}})\\
 &\stackrel{(b)}{\geq}H(F|M_{\mathcal{T}}, K_{\mathcal{T}})-H(F|S_{\mathcal{S}\cup \mathcal{T}})+H(M_{\mathcal{S}},K_{\mathcal{S}}|F,M_{\mathcal{T}}, K_{\mathcal{T}})\\
  &\stackrel{(c)}{\geq}H(F|M_{\mathcal{T}}, K_{\mathcal{T}})-H(F|\hat F({\mathcal{S}}\cup \mathcal{T}))+H(M_{\mathcal{S}},K_{\mathcal{S}}|F,M_{\mathcal{T}}, K_{\mathcal{T}})\\
  &\stackrel{}{=} H(F)-I(F;M_{\mathcal{T}},K_{\mathcal{T}})-H(F|\hat F({\mathcal{S}}\cup \mathcal{T}))+H(M_{\mathcal{S}},K_{\mathcal{S}}|F,M_{\mathcal{T}}, K_{\mathcal{T}})\\
  &\stackrel{(d)}{\geq}H(F)-I(F;M,K_{\mathcal{T}})-H(F|\hat F({\mathcal{S}}\cup \mathcal{T}))+H(M_{\mathcal{S}},K_{\mathcal{S}}|F,M_{\mathcal{T}}, K_{\mathcal{T}})\\
  &\stackrel{(e)}{\geq}H(F)-\alpha_{d} \times H(F)+H(M_{\mathcal{S}},K_{\mathcal{S}}|F,M_{\mathcal{T}}, K_{\mathcal{T}})\\
  &\stackrel{}{=}(1-\alpha_{d})H(F)+H(M_{\mathcal{S}},K_{\mathcal{S}}|F,M_{\mathcal{T}}, K_{\mathcal{T}})\\
  &\stackrel{(f)}{\geq}(1-\alpha_{d})H(F)+H(M_{\mathcal{S}},K_{\mathcal{S}}|F,R, K_{\mathcal{T}})\\
  &\geq (1-\alpha_{d})H(F)+H(K_{\mathcal{S}}|F,R, K_{\mathcal{T}})\\
  &\stackrel{(g)}{=}(1-\alpha_{d})H(F)+\sum_{l\in \mathcal{S}}H(K_{l}),
\end{align*}
where 
    \begin{enumerate}[(a)]
        \item 
  holds because conditioning reduces entropy; \item  holds by the data processing inequality because $S_{\mathcal{S}\cup \mathcal{T}}$ is a function of $(M_{\mathcal{S}},K_{\mathcal{S}},M_{\mathcal{T}},K_{\mathcal{T}})$; \item  holds by Definition~\ref{def1}, and the data processing inequality because $\hat F(\mathcal{{\mathcal{S}\cup \mathcal{T}}})$ is a function of $S_{\mathcal{{\mathcal{S}\cup \mathcal{T}}}}$; \item  holds because  $M$ contains $M_{\mathcal{T}}$; \item for $d\in \llbracket1,D\rrbracket$ holds by~\eqref{marker1b} and \eqref{marker2b}, otherwise holds by~\eqref{marker1b} and \eqref{marker3b}  for $d=0$, because $\alpha_0 \triangleq 0$ and $z_0 \triangleq 0$; \item  holds by the data processing inequality because $ M_{\mathcal{T}}$  is a function of  $( F,R,K_{\mathcal{T}})$; \item holds because the keys $K_{\mathcal{S}}$ are jointly independent and $K_{\mathcal{S}}$ is independent from $(F,R, K_{\mathcal{T}})$.\end{enumerate}\end{proof}
    Then, we have
\begin{align}
   \frac{L}{t-z_{d}}r^{(F)}
     \nonumber &\stackrel{(a)}{=}\frac{L}{t-z_{d}}H(F)\\
    \nonumber  &\stackrel{(b)}{=}\gamma\displaystyle \sum_{ \substack{ {\mathcal{T} \subseteq \mathcal{L}}\\ {|\mathcal{T}|=z_{d}}}} \displaystyle \sum_{\substack{{\mathcal{S} \subseteq \mathcal{T}^c}\\{|\mathcal{S}|=t-z_{d}}}}H(F)\\
     \nonumber &\stackrel{(c)}{\leq} \gamma\frac{1}{1-\alpha_{d}} \displaystyle \sum_{ \substack{ {\mathcal{T} \subseteq \mathcal{L}}\\ {|\mathcal{T}|=z_{d}}}}\displaystyle \sum_{\substack{{\mathcal{S} \subseteq \mathcal{T}^c}\\{|\mathcal{S}|=t-z_{d}}}} \sum_{{l \in \mathcal{S}}}H(M_{l})\\
     \nonumber &\stackrel{(d)}{=}\gamma\frac{1}{1-\alpha_{d}}\displaystyle \sum_{ \substack{ {\mathcal{T} \subseteq \mathcal{L}}\\ {|\mathcal{T}|=z_{d}}}} \binom{L-z_{d}-1}{t-z_{d}-1}\sum_{{l \in \mathcal{T}^c}}H(M_{l})\\
     \nonumber &\stackrel{(e)}{=}\gamma\frac{1}{1-\alpha_{d}}\binom{L-z_{d}-1}{t-z_{d}-1}\displaystyle \sum_{ \substack{ {\mathcal{T} \subseteq \mathcal{L}}\\ {|\mathcal{T}|=L-z_{d}}}}\sum_{{l \in \mathcal{T}}}H(M_{l})\\
     \nonumber &\stackrel{(f)}{=}\gamma\frac{1}{1-\alpha_{d}}\binom{L-z_{d}-1}{t-z_{d}-1}\binom{L-1}{L-z_{d}-1} \sum_{{l \in \mathcal{L}}}H(M_{l})\\
   \nonumber &\stackrel{(g)}{=}\frac{1}{1-\alpha_{d}} \sum_{{l \in \mathcal{L}}}H(M_{l})\\
    &\stackrel{(h)}{\leq}\frac{1}{1-\alpha_{d}} \sum_{{l \in \mathcal{L}}}r_{l}^{(M)}, \label{marker17}
\end{align}
where
     \begin{enumerate}[(a)]
        \item  holds by uniformity of $F$; \item holds with $\gamma\triangleq~\frac{L}{t-z_{d}}\binom L{z_{d}}^{-1}\binom {L-z_{d}}{t-z_{d}}^{-1}$; \item holds by \eqref{marker16}; \item and (f) hold by \cite[Lemma 3.2]{b21}; \item holds by a change of variable in the sum; \addtocounter{enumi}{1} \item  holds because $\gamma\binom{L-z_{d}-1}{t-z_{d}-1}\binom{L-1}{L-z_{d}-1}=1$; \item  holds by the definition of $\mathcal{M}_{l}$.\end{enumerate} 

Finally, since \eqref{marker17} is valid for any private file storage strategy, it is also valid for a $(2^{r^{(F)}},2^{r^{(R)}},(2^{r^{(M)}_{l}})_{l \in \mathcal{L}}, (2^{r^{(S)}_{l}})_{l \in \mathcal{L}}) $ private file storage
strategy that $(\boldsymbol{\alpha}, t, \boldsymbol{z})$-achieves $r^{(F)}$, where $r^{(F)}=r^{(F)}_{*}(\boldsymbol{\alpha}, t, \boldsymbol{z})$,    and thus
\begin{align*}
    &\forall d \in \llbracket 0, D\rrbracket, L  \frac{1-\alpha_{d}}{t-z_{d}}r^{(F)}_{*}(\boldsymbol{\alpha}, t, \boldsymbol{z}) \leq r^{(M)}_{\sum,*} (\boldsymbol{\alpha}, t, \boldsymbol{z}) .
\end{align*}
\section{Proof of Theorem~\ref{thm4b}}\label{appendix:d}Consider an arbitrary $(2^{r^{(F)}},2^{r^{(R)}},(2^{r^{(M)}_{l}})_{l \in \mathcal{L}}, (2^{r^{(S)}_{l}})_{l \in \mathcal{L}}) $ private file storage  strategy that $(\boldsymbol{\alpha}, t, \boldsymbol{z})$-achieves $r^{(F)}$ and suppose that  \eqref{marker10} holds. 

For $i \in \mathcal{L}$, and $\mathcal{S} \subseteq \mathcal{L}$ such that $|\mathcal{S}|=i$, define $\beta_{i}= I (F;M_{\mathcal{S}},K_{\mathcal{S}})$ and $\beta_{L+1} \triangleq \beta_{L}$. Note that $\beta_{i}$ only depends on $i$ and not on $\mathcal{S}$ by \eqref{marker10}. We will need the following lemma. 
\begin{lem} \label{lem3}
Let $ d\in \llbracket0,D\rrbracket$ and $\mathcal{F}$ be the set of all the functions $f$:$\llbracket1,t-z_{d}+2\rrbracket \xrightarrow{} [0,H(F)]$ such that $f(1)=\beta_{z_{d}}$, $f(t-z_{d}+2)=f(t-z_{d}+1)=\beta_{t}$. Then, we have
   $$ H(M_{l})\geq  \displaystyle  \min_{f\in \mathcal{F}} \sum_{i=1}^{t-z_{d}}[2f(i+1)-f(i)-f(i+2)]^+ .$$\end{lem}
   \begin{proof}
  Fix $l \in \mathcal{L}$. For $i \in [z_{d}, t-1]$, define \begin{align*}
\mathcal{V}_{i} &\triangleq
\begin{cases}
\llbracket 1,i \rrbracket, & \text{if } l > i,\\[4pt]
\llbracket 1,i+1 \rrbracket \setminus \{l\}, & \text{if } l \le i,
\end{cases} \\
\mathcal{V}_{t} &\triangleq \mathcal{V}_{t-1} \cup \{l\}.
\end{align*}  
Then, we have
\begin{align}
    H(M_{l})+H(K_{l})
    \nonumber& \geq  H(M_{l},K_{l})\\
   \nonumber& \stackrel{(a)}{\geq}H(M_{l},K_{l}|M_{\mathcal{V}_{z_{d}}},K_{\mathcal{V}_{z_{d}}})\\
   \nonumber& \stackrel{(b)}{=}H(M_{l},K_{l}|M_{\mathcal{V}_{z_{d}}},K_{\mathcal{V}_{z_{d}}})-H(M_{l},K_{l}|M_{\mathcal{V}_{t}},K_{\mathcal{V}_{t}})\\
     & = \displaystyle \sum_{i=z_{d}}^{t-1}[H(M_{l},K_{l}|M_{\mathcal{V}_{i}},K_{\mathcal{V}_{i}})-H(M_{l},K_{l}|M_{\mathcal{V}_{i+1}},K_{\mathcal{V}_{i+1}})] \label{marker19}\\
     \nonumber& \stackrel{(c)}{=} \displaystyle \sum_{i=z_{d}}^{t-1}[H(M_{l},K_{l},F|M_{\mathcal{V}_{i}},K_{\mathcal{V}_{i}})-H(F|M_{\mathcal{V}_{i\cup \{l\}}},K_{\mathcal{V}_{i\cup \{l\}}})\\ \nonumber&-H(M_{l},K_{l},F|M_{\mathcal{V}_{i+1}},K_{\mathcal{V}_{i+1}})\\ \nonumber&+H(F|M_{\mathcal{V}_{i+1\cup \{l\}}},K_{\mathcal{V}_{i+1\cup \{l\}}})]\\
     \nonumber& \stackrel{(d)}{=} \displaystyle \sum_{i=z_{d}}^{t-1}[H(F|M_{\mathcal{V}_{i}},K_{\mathcal{V}_{i}})+H(M_{l},K_{l}|F,M_{\mathcal{V}_{i}},K_{\mathcal{V}_{i}})\\ \nonumber&-H(F|M_{\mathcal{V}_{i\cup \{l\}}},K_{\mathcal{V}_{i\cup \{l\}}}) -H(F|M_{\mathcal{V}_{i+1}},K_{\mathcal{V}_{i+1}})\\ \nonumber&-H(M_{l},K_{l}|F,M_{\mathcal{V}_{i+1}},K_{\mathcal{V}_{i+1}})\\ \nonumber&+H(F|M_{\mathcal{V}_{i+1\cup \{l\}}},K_{\mathcal{V}_{i+1\cup \{l\}}})]\\
     \nonumber& \stackrel{}{=}\displaystyle \sum_{i=z_{d}}^{t-1}[-I(F;M_{\mathcal{V}_{i}},K_{\mathcal{V}_{i}})+H(M_{l},K_{l}|F,M_{\mathcal{V}_{i}},K_{\mathcal{V}_{i}})\\ \nonumber&+I(F;M_{\mathcal{V}_{i\cup \{l\}}},K_{\mathcal{V}_{i\cup \{l\}}}) +I(F;M_{\mathcal{V}_{i+1}},K_{\mathcal{V}_{i+1}})\\ \nonumber&-H(M_{l},K_{l}|F,M_{\mathcal{V}_{i+1}},K_{\mathcal{V}_{i+1}})\\ \nonumber&-I(F;M_{\mathcal{V}_{i+1\cup \{l\}}},K_{\mathcal{V}_{i+1\cup \{l\}}})]\\
     \nonumber& \stackrel{(e)}{=} \displaystyle \sum_{i=z_{d}}^{t-1}[2\beta_{i+1}-\beta_{i}-\beta_{i+2}+H(M_{l},K_{l}|F,M_{\mathcal{V}_{i}},K_{\mathcal{V}_{i}})-H(M_{l},K_{l}|F,M_{\mathcal{V}_{i+1}},K_{\mathcal{V}_{i+1}})] \\
    \nonumber& \stackrel{(f)}{\geq} \displaystyle [2\beta_{t}-\beta_{t-1}-\beta_{t+1}+H(M_{l},K_{l}|F,M_{\mathcal{V}_{t-1}},K_{\mathcal{V}_{t-1}})]^+ + \sum_{i=z_{d}}^{t-2}[2\beta_{i+1}-\beta_{i}-\beta_{i+2}]^+  \\
    \nonumber& \stackrel{(g)}{\geq} \displaystyle [\beta_{t}-\beta_{t-1}+H(M_{l},K_{l}|F,M_{\mathcal{V}_{t-1}},K_{\mathcal{V}_{t-1}})]^++ \sum_{i=z_{d}}^{t-2}[2\beta_{i+1}-\beta_{i}-\beta_{i+2}]^+  \\
    \nonumber& \stackrel{(h)}{=} \displaystyle [\beta_{t}-\beta_{t-1}+H(M_{l},K_{l}|F,M_{\mathcal{V}_{t-1}},K_{\mathcal{V}_{t-1}})]+ \sum_{i=z_{d}}^{t-2}[2\beta_{i+1}-\beta_{i}-\beta_{i+2}]^+  \\
    \nonumber& \stackrel{(i)}{=} \displaystyle H(M_{l},K_{l}|F,M_{\mathcal{V}_{t-1}},K_{\mathcal{V}_{t-1}})\!+\!\! \sum_{i=z_{d}}^{t-1}[2\beta_{i+1}\!-\beta_{i}\!-\beta_{i+2}]^+  \\
   \nonumber & \geq \displaystyle H(K_{l}|F,M_{\mathcal{V}_{t-1}},K_{\mathcal{V}_{t-1}})+ \sum_{i=z_{d}}^{t-1}[2\beta_{i+1}-\beta_{i}-\beta_{i+2}]^+  \\
     \nonumber& \stackrel{(j)}{\geq} \displaystyle H(K_{l}|F,R,K_{\mathcal{V}_{t-1}})+ \sum_{i=z_{d}}^{t-1}[2\beta_{i+1}-\beta_{i}-\beta_{i+2}]^+  \\
     \nonumber & \stackrel{(k)}{=} \displaystyle H(K_{l})+ \sum_{i=z_{d}}^{t-1}[2\beta_{i+1}-\beta_{i}-\beta_{i+2}]^+  \\
      & \geq H(K_{l})+ \displaystyle  \min_{f\in \mathcal{F}} \sum_{i=1}^{t-z_{d}}[2f(i+1)-f(i)-f(i+2)]^+ , \label{marker20}
\end{align}
where
     \begin{enumerate}[(a)]
    \item holds because conditioning reduces entropy; \item holds because $l\in \mathcal{V}_{t}$; \item and   
    (d) hold by the chain rule\addtocounter{enumi}{1}; \item holds by the definition of~$\beta_{i}$; \item  holds by \eqref{marker19} and because conditioning reduces entropy; \item  holds by \eqref{marker1b}; \item holds because $\beta_t\geq \beta_{t-1}$; \item holds because $\beta_t-\beta_{t-1}=[ 2\beta_t-\beta_{t-1}-\beta_{t+1}]^+$; \item holds by the data processing inequality; \item holds because $K_{l}$ is independent of $(F, R, K_{\mathcal{V}_{t-1}})$.\end{enumerate}\end{proof} Then, for $f\in \mathcal{F}$, and $d \in \llbracket 0 ,D \rrbracket$, we have
   \begin{align}
      \displaystyle   \sum_{i=1}^{t-z_{d}}[2f(i+1)-f(i)-f(i+2)]^+  
     \nonumber & \stackrel{(a)}\ge \sum_{i=j}^{t-z_{d}}[2f(i+1)-f(i)-f(i+2)]^+ \\ 
  &\ge \sum_{i=j}^{t-z_{d}}(2f(i+1)-f(i)-f(i+2) )  \nonumber \\
&= \sum_{i=j}^{t-z_{d}}\bigl(f(i+1)-f(i)\bigr)-\sum_{i=j}^{t-z_{d}}\bigl(f(i+2)-f(i+1)\bigr) \nonumber\\
&= \bigl(f(j+1)-f(j)\bigr)-\bigl(f(t-z_{d}+2)-f(t-z_{d}+1)\bigr) \nonumber\\
&= f(j+1)-f(j)\nonumber\\
& = \max_{i\in \llbracket 1,t-z_d\rrbracket}\bigl(f(i+1)-f(i)\bigr) \nonumber \\
& \ge \frac{1}{t-z_d}\sum_{i=1}^{t-z_d}\bigl(f(i+1)-f(i)\bigr) \nonumber \\
& =\frac{f(t-z_d+1)-f(1)}{t-z_d}\nonumber \\
&=\frac{\beta_t-\beta_{z_d}}{t-z_d}\nonumber \\
&\ge H(F)\frac{1-\alpha_d}{t-z_d},\label{marker21}
   \end{align}
   where in $(a)$  $j\in \llbracket 1,t-z_d\rrbracket$ is such that $$f(j+1)-f(j)=\max_{i\in \llbracket 1,t-z_d\rrbracket}\bigl(f(i+1)-f(i)\bigr);$$ $(b)$ holds because   $\beta_{z_{d}} \leq \alpha_{d} \times H(F)$ by~\eqref{marker2b} and $\beta_{t} = H(F)$ by~\eqref{marker1b}.

   Finally, for $l \in \mathcal{L}$, we have
   \begin{align}
   \nonumber r^{(M)}_{l} &\geq H(M_l)\\
   & \stackrel{}{\geq} r^{(F)}  \frac{1-\alpha_{d}}{t-z_{d}} , \label{marker22}  
   \end{align}
     where the last inequality holds by the uniformity of $F$, Lemma~\ref{lem3}, and \eqref{marker21}. Since \eqref{marker22} is valid for any private file storage strategy, it is also valid for a $(2^{r^{(F)}},2^{r^{(R)}},(2^{r^{(M)}_{l}})_{l \in \mathcal{L}}, (2^{r^{(S)}_{l}})_{l \in \mathcal{L}}) $ private file storage
strategy that $(\boldsymbol{\alpha}, t, \boldsymbol{z})$-achieves $r^{(F)}$, where $r^{(F)}=r^{(F)}_{*}(\boldsymbol{\alpha}, t, \boldsymbol{z})$,    and thus for $l \in \mathcal{L}$
\begin{align*}
    &\forall d \in \llbracket 0, D\rrbracket,   \frac{1-\alpha_{d}}{t-z_{d}}r^{(F)}_{*}(\boldsymbol{\alpha}, t, \boldsymbol{z}) \leq r^{(M)}_{l,*} (\boldsymbol{\alpha}, t, \boldsymbol{z}) .
\end{align*}

\section{Proof of Theorem~\ref{thm5}}\label{appendix:e}
  
Consider an arbitrary $(2^{r^{(F)}},2^{r^{(R)}},(2^{r^{(M)}_{l}})_{l \in \mathcal{L}}, (2^{r^{(S)}_{l}})_{l \in \mathcal{L}}) $ private file storage  strategy that $(\alpha,t,z)$-achieves $r^{(F)}$. 

Let $\mathcal{V} \subseteq \mathcal{L}$ such that $v \triangleq|\mathcal{V}|< z$. Then, let $\mathcal{T} \subseteq \mathcal{L} \backslash \mathcal{V}$ and $\mathcal{S} \subseteq \mathcal{L} \backslash ( \mathcal{T} \cup \mathcal{V})$ such that$ |\mathcal{T}|=z-v$, and $ |\mathcal{S}|=t-z$. We first derive the following lemma.
\begin{lem}
    We have\label{lem7}
\begin{align*}
    \displaystyle \sum_{l \in \mathcal{S}} H(M_{l},K_{l}|M_{\mathcal{V}}, K_{\mathcal{V}}) \geq (1- \alpha)H(F)+ \displaystyle \sum_{l \in \mathcal{S}}H(K_{l}). 
\end{align*} 
\end{lem}
\begin{proof}
     We have
\begin{align*}
    \displaystyle \sum_{l \in \mathcal{S}} H(M_{l},K_{l}|M_{\mathcal{V}}, K_{\mathcal{V}})
   & \stackrel{(a)}{\geq} H(M_{\mathcal{S}},K_{\mathcal{S}}|M_{\mathcal{V}}, K_{\mathcal{V}})\\
   & \stackrel{(b)}{\geq} H(M_{\mathcal{S}},K_{\mathcal{S}}|M_{\mathcal{V}\cup \mathcal{T}}, K_{\mathcal{V}\cup \mathcal{T}}) \\
   & \stackrel{}{=} I(M_{\mathcal{S}},K_{\mathcal{S}};F|M_{\mathcal{V}\cup \mathcal{T}}, K_{\mathcal{V}\cup \mathcal{T}})+H(M_{\mathcal{S}},K_{\mathcal{S}}|F,M_{\mathcal{V}\cup \mathcal{T}}, K_{\mathcal{V}\cup \mathcal{T}}) \\
   & \stackrel{}{=} H(F|M_{\mathcal{V}\cup \mathcal{T}}, K_{\mathcal{V}\cup \mathcal{T}})-H(F|M_{\mathcal{V}\cup \mathcal{T}\cup \mathcal{S}}, K_{\mathcal{V}\cup \mathcal{T}\cup \mathcal{S}})+H(M_{\mathcal{S}},K_{\mathcal{S}}|F,M_{\mathcal{V}\cup \mathcal{T}}, K_{\mathcal{V}\cup \mathcal{T}})\\
   & \stackrel{(c)}{\geq} H(F|M_{\mathcal{V}\cup \mathcal{T}}, K_{\mathcal{V}\cup \mathcal{T}})-H(F| S_{\mathcal{V}\cup \mathcal{T}\cup \mathcal{S}})+H(M_{\mathcal{S}},K_{\mathcal{S}}|F,M_{\mathcal{V}\cup \mathcal{T}}, K_{\mathcal{V}\cup \mathcal{T}})\\
   & \stackrel{(d)}{\geq} H(F|M_{\mathcal{V}\cup \mathcal{T}}, K_{\mathcal{V}\cup \mathcal{T}})-H(F| \hat{F}(\mathcal{V}\cup \mathcal{T}\cup \mathcal{S}))+H(M_{\mathcal{S}},K_{\mathcal{S}}|F,M_{\mathcal{V}\cup \mathcal{T}}, K_{\mathcal{V}\cup \mathcal{T}}) \\ 
   & \stackrel{(e)}{=} H(F)+H(M_{\mathcal{S}},K_{\mathcal{S}}|F,M_{\mathcal{V}\cup \mathcal{T}}, K_{\mathcal{V}\cup \mathcal{T}})-I(F;M_{\mathcal{V} \cup \mathcal{T}},K_{\mathcal{V} \cup \mathcal{T}}) \displaybreak[1]\\
   & \stackrel{(f)}{\geq}(1- \alpha)H(F)+H(M_{\mathcal{S}},K_{\mathcal{S}}|F,R, K_{\mathcal{V}\cup \mathcal{T}}) \\
   & \stackrel{(g)}{\geq} (1- \alpha)H(F)+H(K_{\mathcal{S}}|F,R, K_{\mathcal{V}\cup \mathcal{T}})\\
   & \stackrel{(h)}{=} (1- \alpha)H(F)+H(K_{\mathcal{S}})\\
   & \stackrel{(i)}{=} (1- \alpha)H(F)+ \displaystyle \sum_{l \in \mathcal{S}}H(K_{l}),
\end{align*}
 where
     \begin{enumerate}[(a)]
    \item and (b) hold by the chain rule and  because conditioning reduces entropy\addtocounter{enumi}{1}; \item 
       holds by the data processing inequality and because $S$ is a function of $(M_{(\mathcal{V}\cup\mathcal{T}\cup\mathcal{S})},K_{(\mathcal{V}\cup\mathcal{T}\cup\mathcal{S})})$; \item holds by the data processing inequality and because $\hat{F}(\mathcal{V}\cup\mathcal{T}\cup\mathcal{S})$ is a function of $S(\mathcal{V}\cup\mathcal{T}\cup\mathcal{S})$; \item  holds by \eqref{marker1} because $|\mathcal{V}\cup \mathcal{T}\cup \mathcal{S}| = t$; \item holds by \eqref{marker2} because $|\mathcal{V}\cup \mathcal{T}| = z$, and by the data  processing inequality because $M_{\mathcal{V} \cup \mathcal{T}}$ is a function of $(F,R,K_{\mathcal{V} \cup \mathcal{T}})$; \item hold by the chain rule and non-negativity of the entropy; \item holds because the keys $K_{\mathcal{S}}$ are independent of $(F, R, K_{\mathcal{V} \cup \mathcal{T}})$; \item holds by independence of $K_{\mathcal{S}}$. \end{enumerate}
\end{proof}
Then, from Lemma \ref{lem7}, we deduce Lemma \ref{lem8}.
\begin{lem} \label{lem8}
    Consider 
    $$l^*(\mathcal{V}) \in \displaystyle \arg\max_{l \in \mathcal{L} \backslash \mathcal{V}}[H(M_{l}, K_{l}|M_{\mathcal{V}}, K_{\mathcal{V}})-H(K_{l})].$$
     We have 
   \begin{align*}
     \nonumber  \frac{1-\alpha}{t-z}H(F) \leq H(M_{\mathcal{L}},K_{\mathcal{L}}|M_{\mathcal{V}},K_{\mathcal{V}})-H(M_{\mathcal{L}},K_{\mathcal{L}}|M_{\mathcal{V}\cup \{ l^*(\mathcal{V})\}},K_{\mathcal{V}\cup \{ l^*(\mathcal{V})\}})-n. 
   \end{align*}
\end{lem}
\begin{proof}
    We have
\begin{align*}
     \frac{1-\alpha}{t-z}H(F)
    & \stackrel{(a)}{=}
    \Omega \displaystyle \sum_{\substack{{\mathcal{T} \subseteq \mathcal{L}\backslash \mathcal{V}}\\{|\mathcal{T}|=z-v}}}
    \sum_{\substack{{\mathcal{S} \subseteq \mathcal{L}\backslash (\mathcal{T} \cup \mathcal{V})}\\{|\mathcal{S}|=t-z}}}(1-\alpha)H(F)\\
    & \stackrel{(b)}{\leq}  \Omega \displaystyle \sum_{\substack{{\mathcal{T} \subseteq \mathcal{L}\backslash \mathcal{V}}\displaybreak[0]\\{|\mathcal{T}|=z-v}}}
    \sum_{\substack{{\mathcal{S} \subseteq \mathcal{L}\backslash (\mathcal{T} \cup \mathcal{V})}\\{|\mathcal{S}|=t-z}}}\sum_{l \in \mathcal{S}}[H(M_{l},K_{l}|M_{\mathcal{V}},K_{\mathcal{V}})-H(K_{l})]\\
    & \stackrel{(c)}{=} \Omega \displaystyle \sum_{\substack{{\mathcal{T} \subseteq \mathcal{L}\backslash \mathcal{V}}\\{|\mathcal{T}|=z-v}}}\binom{L-z-1}{t-z-1} \sum_{l \in \mathcal{L} \backslash(\mathcal{T} \cup \mathcal{V})}[H(M_{l},K_{l}|M_{\mathcal{V}},K_{\mathcal{V}})-H(K_{l})]\\
    & \stackrel{(d)}{=} \Omega \binom{L-z-1}{t-z-1} \displaystyle \sum_{\substack{{\mathcal{T}\subseteq \mathcal{L}\backslash \mathcal{V}}\\{|\mathcal{T}|=L-z}}}\sum_{l \in \mathcal{T}}[H(M_{l},K_{l}|M_{\mathcal{V}},K_{\mathcal{V}})-H(K_{l})]\\
    &\stackrel{(e)}{=} \Omega \binom{L-z-1}{t-z-1}\binom{L-v-1}{L-z-1} \displaystyle \sum_{l \in \mathcal{L}\backslash \mathcal{V}}[H(M_{l},K_{l}|M_{\mathcal{V}},K_{\mathcal{V}})-H(K_{l})]\\
    &\stackrel{(f)}{=}\frac{1}{L-v}\displaystyle \sum_{l \in \mathcal{L}\backslash \mathcal{V}}[H(M_{l},K_{l}|M_{\mathcal{V}},K_{\mathcal{V}})-H(K_{l})]\\
    & \stackrel{(g)}{\leq} H(M_{l^{*}(\mathcal{V})},K_{l^{*}(\mathcal{V})}|M_{\mathcal{V}},K_{\mathcal{V}})-H(K_{l^{*}(\mathcal{V})})\\
    & \stackrel{(h)}{=} H(M_{l^{*}(\mathcal{V})},K_{l^{*}(\mathcal{V})}|M_{\mathcal{V}},K_{\mathcal{V}})-n\\
    & \stackrel{(i)}{=}H(M_{\mathcal{L}},K_{\mathcal{L}}|M_{\mathcal{V}},K_{\mathcal{V}})-H(M_{\mathcal{L}},K_{\mathcal{L}}|M_{\mathcal{V}\cup \{ l^*(\mathcal{V})\}},K_{\mathcal{V}\cup \{ l^*(\mathcal{V})\}})-n,
\end{align*}
where
    \begin{enumerate}[(a)]\item holds with $$\Omega \triangleq \frac{1}{t-z}\binom{L-v}{z-v}^{-1}\binom{L-z}{t-z}^{-1};$$ 
    \item  holds by Lemma~\ref{lem7};  
    \item and (e)  hold by \cite[Lemma 3.2]{b21}; \item  holds by a change of variable in the sums;  \addtocounter{enumi}{1}
    \item  holds because $$\Omega \binom{L-z-1}{t-z-1}\binom{L-v-1}{L-z-1}=\frac{1}{L-v};$$ \item holds by the definition of~$l^*(\mathcal{V})$; \item holds because the keys $K_{l^*(\mathcal{V})}$ are uniformly distributed;\item
    holds by the chain rule.\end{enumerate}
\end{proof}
Next, we apply multiple times Lemma~\ref{lem8} to obtain the following lemma.
    \begin{lem}\label{lem9}
        Define $\mathcal{V}_{0}\triangleq \emptyset$ and, for $i \in \llbracket  1 , z \rrbracket$,  $\mathcal{V}_{i} \triangleq \mathcal{V}_{i-1} \cup\{ l^*(\mathcal{V}_{i-1})\}$. We have
        $$ \frac {z(1-\alpha)}{t-z}H(F) \leq H(M_{\mathcal{L}},K_{\mathcal{L}})-H(M_{\mathcal{L}},K_{\mathcal{L}}|M_{\mathcal{V}_z},K_{\mathcal{V}_z})-nz.$$
    \end{lem}
    \begin{proof}
        We have 
\begin{align*}
          \frac {z(1-\alpha)}{t-z}H(F) 
         & \stackrel{}{=}  \displaystyle \sum_{i=0}^{z-1}\frac{1-\alpha}{t-z}H(F) \\
        & \leq \displaystyle \sum_{i=0}^{z-1}[  H(M_{\mathcal{L}},K_{\mathcal{L}}|M_{\mathcal{V}_{i}},K_{\mathcal{V}_{i}})-H(M_{\mathcal{L}},K_{\mathcal{L}}|M_{\mathcal{V}_{i+1}},K_{\mathcal{V}_{i+1}})-n]\\
        & =  H(M_{\mathcal{L}},K_{\mathcal{L}})-H(M_{\mathcal{L}},K_{\mathcal{L}}|M_{\mathcal{V}_{z}},K_{\mathcal{V}_{z}})-nz,
\end{align*}
where the inequality holds by  applying $z$ times Lemma~\ref{lem8} and the definition of $\mathcal{V}_{i} $,  $ i \in \llbracket 0, z \rrbracket$.\end{proof}
Finally, we simplify the upper bound of Lemma~\ref{lem9} as follows. We have
\begin{align}
     \nonumber\frac {z(1-\alpha)}{t-z}r^{(F)} \displaybreak[0]
      \nonumber& \stackrel{(a)}{=} \frac {z(1-\alpha)}{t-z}H(F)\\
      \nonumber& \stackrel{(b)}{\leq} H(M_{\mathcal{L}},K_{\mathcal{L}})-H(M_{\mathcal{L}},K_{\mathcal{L}}|M_{\mathcal{V}_{z}},K_{\mathcal{V}_{z}})-nz\\
      \nonumber& \stackrel{(c)}{=}H(M_{\mathcal{L}},K_{\mathcal{L}})-H(F,M_{\mathcal{L}},K_{\mathcal{L}}|M_{\mathcal{V}_{z}},K_{\mathcal{V}_{z}})-nz\\
       \nonumber& \stackrel{(d)}{\leq}H(M_{\mathcal{L}},K_{\mathcal{L}})-H(F,K_{\mathcal{V}^c_{z}}|M_{\mathcal{V}_{z}},K_{\mathcal{V}_{z}})-nz\\
      \nonumber & \stackrel{(e)}{=}H(M_{\mathcal{L}},K_{\mathcal{L}})-H(F|M_{\mathcal{V}_{z}},K_{\mathcal{V}_{z}})- H(K_{\mathcal{V}^c_{z}}|F,M_{\mathcal{V}_{z}},K_{\mathcal{V}_{z}})-nz\\
       \nonumber& \stackrel{}{\leq}H(M_{\mathcal{L}},K_{\mathcal{L}})-(1- \alpha) H(F)- H(K_{\mathcal{V}^c_{z}}|F,M_{\mathcal{V}_{z}},K_{\mathcal{V}_{z}})-nz\\
       \nonumber & \stackrel{(f)}{\leq} H(M_{\mathcal{L}},K_{\mathcal{L}})-H(K_{\mathcal{V}^c_{z}}|F,R,K_{\mathcal{V}_{z}})-(1-\alpha) H(F)-nz \\
       \nonumber & \stackrel{(g)}{=}H(M_{\mathcal{L}},K_{\mathcal{L}})-H(K_{\mathcal{V}^c_{z}})-(1-\alpha) H(F)-nz\\
         \nonumber& \stackrel{(h)}{=} H(M_{\mathcal{L}},K_{\mathcal{L}})-(1-\alpha) H(F)-nL\\
          \nonumber& \stackrel{(i)}{\leq}H(R,F,K_{\mathcal{L}})-(1-\alpha) H(F)-nL\\
           \nonumber& = H(R)+H(K_{\mathcal{L}})+H(F)-(1-\alpha) H(F)-nL\\
            \nonumber & \stackrel{(j)}{=} H(R)+H(F)-(1-\alpha) H(F)\\
              & \stackrel{(k)}{=} r^{(R)} +\alpha r^{(F)}, \nonumber
 \end{align}
 where
    \begin{enumerate}[(a)]\item   holds by uniformity of $F$; \item holds by  Lemma~\ref{lem9}; \item   holds by Definition~\ref{def1}; \item   holds by  the chain rule and non-negativity of entropy; \item   holds by  the chain rule; \item holds by the data processing inequality because $M_{\mathcal{V}_{z}}$ is a function of $(F,R,K_{\mathcal{V}_{z}})$; \item holds because the keys are jointly independent and $K_{\mathcal{V}^c_{z}}$ is independent of $(F, R, K_{\mathcal{V}_{z}})$; \item holds because $K_{\mathcal{V}^c_{z}}$  is uniformly distributed; \item  holds by the data processing inequality because  $M_{\mathcal{L}}$ is a function of $(F,R,K_{\mathcal{L}})$; \item  holds because  $K_{\mathcal{L}}$ is  distributed uniformly; \item holds because  $R$ and $F$ are uniformly distributed. 
\end{enumerate}\bibliographystyle{IEEEtran}
          \bibliography{bib}
 \end{document}